\pdfoutput=1
\documentclass[conference]{IEEEtran}
\IEEEoverridecommandlockouts
\usepackage{amsmath,graphicx,algorithmic,algorithm}
\usepackage{array}
\usepackage{amsfonts}
 \DeclareMathOperator*{\argmin}{argmin}

\newtheorem{proposition}{\textbf{Proposition}}

\renewcommand{\l}{\ell}
\newcommand{\norm}[1]{\lVert#1\rVert}
\newcommand{\card}[1]{\lvert#1\rvert}

\newcommand{\R}{{\mathbb{R}}}

\newcommand{\Rbb}{\mathbb{R}}

\renewcommand{\L}{{\mathcal{L}}}

\newcommand{\Scal}{\mathcal{S}}

\newcommand{\transpose}{{\!\scriptscriptstyle\mathrm T}}  

\def\L{{\cal L}}

\title{Chebyshev Polynomial Approximation for Distributed Signal Processing}
\author{\IEEEauthorblockN{David I Shuman, Pierre Vandergheynst, and Pascal Frossard}\thanks{This work was supported in part by FET-Open grant number 255931 UNLocX. 
The authors would also like to thank Javier P{\'e}rez-Trufero for his help producing some of the graphics in this paper.}
\IEEEauthorblockA{Ecole Polytechnique F{\'e}d{\'e}rale de Lausanne (EPFL) \\ Signal Processing Laboratory \\ CH-1015 Lausanne, Switzerland \\ \{david.shuman,~pierre.vandergheynst,~pascal.frossard\}@epfl.ch}}

\begin{document}
%
\maketitle
\begin{abstract}
Unions of graph Fourier multipliers are an important class of linear operators for processing signals defined on graphs. We present a novel method to efficiently distribute the application of these operators to the high-dimensional signals collected by sensor networks. The proposed method
features approximations of the graph Fourier multipliers by shifted Chebyshev polynomials, whose recurrence relations make them readily amenable to distributed computation. We demonstrate how the proposed method can be used in a distributed denoising task, and 
show that the communication requirements of
the method scale gracefully with the size of the network. 
\end{abstract}
\begin{keywords}
Chebyshev polynomial approximation, denoising, distributed optimization, regularization, signal processing on graphs, spectral graph theory, wireless sensor networks
\end{keywords}
\section{Introduction}\label{sec:intro}

Wireless sensor networks are now prevalent in applications such as environmental 
monitoring, target
tracking, surveillance, medical diagnostics, and manufacturing process flow. The sensor nodes are often deployed \emph{en masse} to collectively achieve tasks such as estimation, detection, classification, and localization.
While such networks have the ability to collect large amounts of data in a short time, they also face a number of resource constraints. First, they are energy constrained, as they are often expected to operate for long periods of time without human intervention, despite being powered by batteries or energy harvesting. Second, they may have limited communication range and 
capacity due to the need to save energy. Third, they may have limited on-board processing capabilities. Therefore, it is critical to develop distributed algorithms for in-network data processing that help balance the trade-offs between performance, communication bandwidth, and computational complexity.

Due to the limited communication range of wireless sensor nodes, 
 each sensor node in a large network is likely to communicate with only a small number of other nodes in the network. To model the communication patterns, we can write down a graph with each vertex corresponding to a sensor node and each edge corresponding to a pair of nodes that communicate. Moreover, because the communication graph is a function of the distances between nodes, it often captures spatial correlations between sensors' observations as well.
 That is, if two sensors are close enough to communicate, their observations are more likely to be correlated.
We can further specify these 
spatial correlations by adding weights to the edges of the 
graph, with higher weights associated to edges connecting sensors with closely correlated observations. For example,
it is common to construct the graph with
a thresholded Gaussian kernel weighting function based on the physical distance between nodes, where the weight of edge $e$ connecting nodes $i$ and $j$ that are a distance $d(i,j)$ 
apart is 
\begin{eqnarray}\label{Eq:gkw}
w(e)=
\begin{cases}
\exp\left({-\frac{[d(i,j)]^2}{2\sigma^2}}\right) &\mbox{if } d(i,j) \leq \kappa \\
0 &\mbox{otherwise}
\end{cases},
\end{eqnarray}
for some parameters $\sigma$ and $\kappa$.

In this paper, we consider signals collected by a sensor network whose nodes can only send messages to their local neighbors
 (i.e., 
they cannot communicate directly with a central entity). While much 
of the literature on distributed signal processing (see, e.g., \cite{Rabbat}-\nocite{predd, olfati}\cite{dimakis} and references therein) focuses 
on 
coming to an agreement on simple features of the observed signal (e.g., consensus averaging, parameter estimation), we are more interested in processing the full function in a distributed manner, with each node having its own objective.
%
%
Some example tasks under this umbrella include:
\begin{itemize}
\item \emph{Distributed denoising} -- In a sensor network of $N$ sensors, a noisy $N$-dimensional signal is observed, with each component of the signal corresponding to the observation at one sensor location. Using the prior knowledge that the denoised signal should be smooth or piecewise smooth with respect to the underlying weighted graph structure, the sensors' task is to denoise each of their components of the signal by iteratively passing messages to their local neighbors and performing computations.
\item \emph{Distributed semi-supervised learning / binary classification} -- A binary label (-1 or 1) is associated with each sensor node; however, only a small number of nodes in the network have knowledge of their labels. The cooperative task is for each node to learn its label by iteratively passing messages to its local neighbors and performing computations.
\end{itemize}

These and 
similar tasks have been considered in centralized settings in the relatively young field of signal processing on graphs. For example, \cite{smola}-\nocite{zhou_scholkopf}\cite{reg_discrete} consider general regularization frameworks on weighted graphs; \cite{harmonic}-\nocite{belkin_matveeva}\cite{zhou_bousquet} present graph-based semi-supervised learning methods; and \cite{bougleux}-\nocite{elmoataz, hancock}\cite{peyre_nlr} consider regularization and filtering on weighted graphs for image and mesh processing.
In distributed settings, \cite{wagner} considers denoising via wavelet processing and \cite{barbarossa} presents a denoising algorithm that projects the measured signal onto a low-dimensional subspace spanned by smooth functions.
References \cite{guestrin}-\nocite{predd_tit}\cite{dlasso} consider different distributed regression problems.

Our main contributions in this paper are i) to show that a key component of many distributed signal processing tasks is the application of linear operators that are unions of graph Fourier multipliers; and ii) to present a novel method to efficiently distribute the application of the graph Fourier multiplier operators to the high-dimensional signals collected by sensor networks.

To elaborate a bit, graph Fourier multiplier operators are the graph analog of filter banks, one of the most commonly used tools in digital signal processing. Multiplying a signal on the graph by one of these matrices is analogous to reshaping the signal's frequencies by multiplying it by a filter in the Fourier domain in classical signal processing. The crux of our novel distributed computational method is to approximate each graph Fourier multiplier by a truncated Chebyshev polynomial expansion. In a centralized setting, \cite{LTS-ARTICLE-2009-053} shows that the truncated Chebyshev polynomial expansion efficiently approximates the application of a spectral graph wavelet transform, which is a specific example of a union of graph Fourier multipliers.
In this paper, we extend the Chebyshev polynomial approximation method to the general class of unions of graph Fourier multiplier operators, and show how the recurrence properties of the Chebyshev polynomials also enable distributed application of these operators. The communication requirements for distributed computation using this method scale gracefully with the number of sensors in the network (and, accordingly, the size of the signals).

The remainder of the paper is as follows. In the next section, we provide 
some background from
spectral graph theory. In Section \ref{Se:GFM}, we introduce graph Fourier multiplier operators and show how they can be efficiently approximated with shifted Chebyshev polynomials in a centralized setting. We then discuss the distributed computation of quantities involving these operators in Section \ref{Se:distribution}, and provide some application examples in Section \ref{Se:applications}. Section \ref{Se:conclusion} concludes the paper.

\section{Spectral Graph Theory} \label{Se:SGT}
Before proceeding, we introduce some basic notations and definitions from spectral graph theory \cite{chung}. We model the sensor network with an undirected, weighted graph $G = \{E,V,w\}$, which consists of a set of vertices $V$, a
set of edges $E$, and a weight function $w:E\to\mathbb{R}^+$ that assigns a
non-negative weight to each edge. We assume the number of sensors in the network, $N=|V|$, is finite, and the graph is connected. The adjacency (or weight) matrix $A$ for a weighted graph $G$ is the $N\times N$ matrix with entries $A_{m,n}$, where
\begin{equation*}
A_{m,n} =
\begin{cases}
w(e), &\mbox{ if $e\in E$ connects vertices $m$ and $n$} \\
0, &\mbox{ otherwise}
\end{cases}~.
\end{equation*}
The degree of each vertex is the sum of the weights of all the edges incident
to it. We define the degree matrix $D$
to have diagonal elements equal to the degrees, and zeros elsewhere.
The non-normalized graph Laplacian is defined as $\L:=D-A$. 
For any $f\in \mathbb{R}^N$ on the vertices of the graph, $\L$ satisfies
\begin{equation*}
(\L f)(m) = \sum_{m\sim n} A_{m,n} \cdot [f(m) - f(n)],
\end{equation*}
where $m\sim n$ indicates 
vertices $m$ and $n$ are connected. 

As the graph Laplacian $\L$ is a real symmetric
matrix, it has a complete set of orthonormal eigenvectors. We denote
these by $\chi_{\l}$ for $\l=0, \hdots ,N-1$, with associated real, non-negative eigenvalues
$\lambda_{\l}$ satisfying $\L \chi_{\l} = \lambda_{\l} \chi_{\l}$.
Zero appears as an eigenvalue with multiplicity equal to the
number of connected components of the graph
\cite{chung}.
%
%
Without loss of generality, we assume the eigenvalues of the Laplacian of the connected graph $G$ to be ordered as
\begin{equation*}
0=\lambda_0 < \lambda_1 \leq \lambda_2 ... \leq \lambda_{N-1}:=\lambda_{\max}.
\end{equation*}

Just as the classical Fourier transform is the expansion of a function $f$ in terms of the eigenfunctions of the Laplace operator
\begin{equation*}
\hat{f}(\omega)= \langle e^{i \omega x},f\rangle = \int\limits_{\Rbb}f(x)  e^{-i\omega x}~dx,
\end{equation*}
the \emph{graph Fourier transform} $\hat{f}$ of any function $f\in\mathbb{R}^{N}$ on the vertices of $G$
is the expansion of $f$ in terms of the eigenfunctions of the graph Laplacian. It is defined by
\begin{equation}\label{Eq:graph_FT}
\hat{f}(\l) := \langle \chi_{\l},f\rangle = \sum_{n=1}^N \chi^*_{\l}(n) f(n),
\end{equation}
where we adopt the convention that the inner product be
  conjugate-linear in the first argument. The \emph{inverse graph Fourier transform} is given by
\begin{equation}\label{Eq:graph_IFT}
f(n) = \sum_{\l=0}^{N-1} \hat{f}(\l) \chi_{\l}(n).
\end{equation}

\section{Chebyshev Polynomial Approximation of Graph Fourier Multipliers} \label{Se:GFM}
In this section, we introduce graph Fourier multiplier operators, unions of graph Fourier multiplier operators, and a computationally efficient approximation to unions of graph Fourier multiplier operators based on shifted Chebyshev polynomials.
All 
methods discussed here 
are for
a centralized setting, and we extend them to a distributed setting in Section \ref{Se:distribution}.


\subsection{Graph Fourier Multiplier Operators} \label{Se:gfmo}
For a function $f$ defined on the real line,
a \emph{Fourier multiplier operator} or \emph{filter} $\Psi$ reshapes the function's
frequencies through multiplication in the Fourier domain:
\begin{equation*}
\widehat{\Psi f}(\omega) = g(\omega)\hat{f}(\omega),\hbox{ for every frequency } \omega.
\end{equation*}
Equivalently,
denoting the Fourier and inverse Fourier transforms by ${\cal F}$ and ${\cal F}^{-1}$, we have
\begin{align}\label{Eq:multiplier_def}
\Psi f(x) &= {\cal F}^{-1}\Bigl(g(\omega){\cal F}(f)(\omega)\Bigr)(x) \\
&= \frac{1}{2\pi}\int\limits_{\Rbb}g(\omega) \hat{f}(\omega) e^{i\omega x}~d\omega. \nonumber
\end{align}

We can extend this straightforwardly to functions defined on the vertices of a graph (and in fact to any group with a Fourier transform) 
by replacing the Fourier transform and its inverse in \eqref{Eq:multiplier_def} with the graph Fourier transform and its inverse, defined in \eqref{Eq:graph_FT} and \eqref{Eq:graph_IFT}. 
Namely, a \emph{graph Fourier multiplier operator} is a linear operator $\Psi: \Rbb^N \rightarrow \Rbb^{N}$ that can be written as
\begin{align}\label{Eq:gfm_def}
\Psi f(n) &= {\cal F}^{-1}\Bigl(g(\lambda_{\l}){\cal F}(f)(\l)\Bigr)(n) \nonumber \\
&= \sum\limits_{\l=0}^{N-1}g(\lambda_{\l}) \hat{f}(\l) \chi_{\l}(n).
\end{align}
We refer to $g(\cdot)$ as the \emph{multiplier}. 
A high-level intuition behind \eqref{Eq:gfm_def} is as follows. The eigenvectors corresponding to the lowest eigenvalues of the graph Laplacian are the ``smoothest'' in the sense that $\left|\chi_{\l}(m)-\chi_{\l}(n)\right|$ is small for neighboring vertices $m$ and $n$. At the extreme is $\chi_{0}$, which is a constant vector ($\chi_{0}(m)=\chi_{0}(n)$ for all $m$ and $n$). The inverse graph Fourier transform \eqref{Eq:graph_IFT} provides a representation of a signal $f$ as a superposition of the orthonormal set of eigenvectors of the graph Laplacian.
The effect of
the graph Fourier multiplier operator $\Psi$ is to modify
the contribution of each eigenvector. 
For example, applying a multiplier $g(\cdot)$ that is 1 for all $\lambda_{\l}$ below some threshold, and 0 for all $\lambda_{\l}$ above the threshold is equivalent to projecting the signal onto the eigenvectors of the graph Laplacian associated with the lowest 
eigenvalues. 
This is analogous to low-pass filtering in the continuous domain. Section \ref{Se:applications} contains further intuition about and examples of graph Fourier multiplier operators. For more properties of the graph Laplacian eigenvectors, see \cite{lap_eigen}.

\subsection{Unions of Graph Fourier Multiplier Operators}
In order for our distributed computation method of the next section to be applicable to a wider range of applications, we can generalize slightly from graph Fourier multipliers to
\emph{unions of graph Fourier multiplier operators}. A union of graph Fourier multiplier operators is a linear operator
$\Phi: \Rbb^N \rightarrow \Rbb^{\eta N}$ ($\eta \in \{1,2,\ldots\}$) whose application to a function
$f \in \Rbb^N$ can be written as (see also Figure \ref{Fig:union})
\begin{align*}
\Phi f &=
\left[
\Psi_1; \Psi_2; \ldots; \Psi_{\eta}
\right] f \\
&=\left[(\Psi_1 f)_1;\ldots;(\Psi_1 f)_N;\ldots;(\Psi_{\eta} f)_1;\ldots;(\Psi_{\eta} f)_N
\right] \\
&=
 \left[
(\Phi f)_1;
(\Phi f)_2;
\ldots;
(\Phi f)_{\eta N}
\right], 
\end{align*}
where for every $j$, $\Psi_j: \Rbb^N \rightarrow \Rbb^{N}$ is a graph Fourier multiplier operator with multiplier 
$g_j(\cdot)$, and
\begin{align} \label{Eq:operator_def}
&\left(\Phi f\right)_{(j-1)N+n} = \sum_{\l=0}^{N-1}g_j(\lambda_{\l}) \hat{f}(\l) \chi_{\l}(n), \\
& \hspace{2.2cm}\hbox{for }j\in\{1,2,\ldots,\eta\},~ n\in\{1,2,\ldots,N\}. \nonumber
\end{align}
\begin{figure}
\centering{
\includegraphics[width=3.2in]{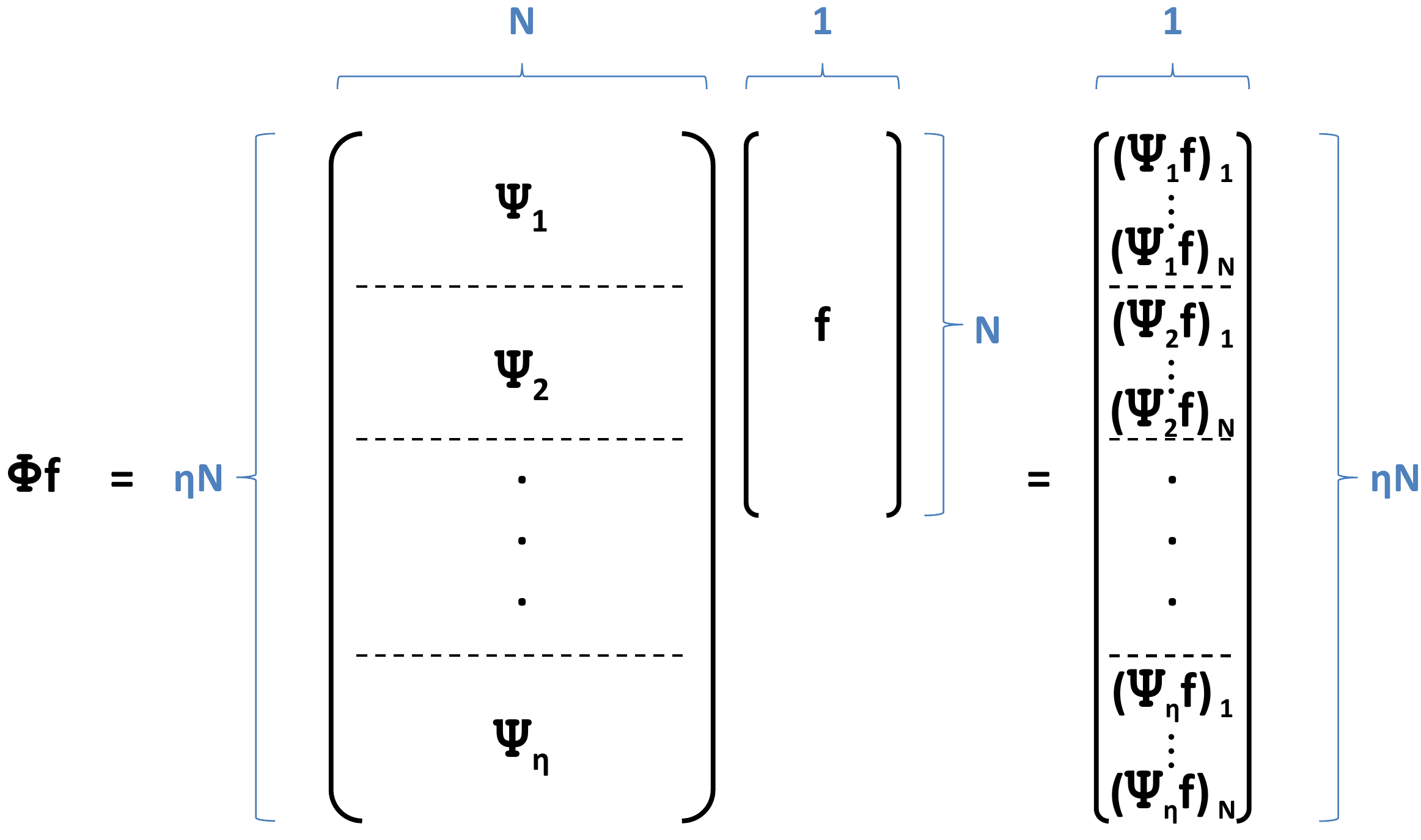}
}\caption{Application of a union of graph Fourier multiplier operators.} \label{Fig:union}
\end{figure}

\subsection{The Chebyshev Polynomial Approximation}
Exactly computing $\Phi f$ requires explicit computation of the entire set of eigenvectors and eigenvalues of $\L$, which becomes computationally challenging as the size of the network, $N$, increases, even in a centralized setting. As discussed in detail in \cite[Section 6]{LTS-ARTICLE-2009-053}, a computationally efficient approximation $\tilde{\Phi}f$ of $\Phi f$ can be computed by approximating each 
multiplier $g_j(\cdot)$ by a truncated series of shifted Chebyshev polynomials. Doing so circumvents the need to compute the full set of eigenvectors and eigenvalues of $\L$. We summarize this approach below.

For $y \in [-1,1]$, the Chebyshev polynomials $\left\{T_k(y)\right\}_{k=0,1,2,\ldots}$ are generated by
\begin{eqnarray*}
T_k(y):=\begin{cases}
1,&\hbox{ if }k=0 \\
y,&\hbox{ if }k=1 \\
2yT_{k-1}(y)-T_{k-2}(y),&\hbox{ if }k\geq 2
\end{cases}.
\end{eqnarray*}
These Chebyshev polynomials form an orthogonal basis for \\
$L^2\left([-1,1],\frac{dy}{\sqrt{1-y^2}} \right)$. So 
every function $h$ on $[-1,1]$ that is square integrable with respect to the measure $dy/\sqrt{1-y^2}$ can be represented as
$h(y)=\frac{1}{2}b_0+\sum_{k=1}^{\infty}b_k T_k(y)$,
where $\{b_k\}_{k=0,1,\ldots}$ is a sequence of Chebyshev coefficients that depends on $h(\cdot)$.
For a detailed overview of Chebyshev polynomials, including the above definitions and properties,
see \cite{handscomb}-\nocite{phillips}\cite{rivlin}.

By shifting the domain of the Chebyshev polynomials to $[0,\lambda_{\max}]$ via the transformation $x=\frac{\lambda_{\max}}{2}(y+1)$, we can represent each 
multiplier as
\begin{eqnarray} \label{Eq:multiplier_expansion}
g_j(x)=\frac{1}{2}c_{j,0}+\sum_{k=1}^{\infty}c_{j,k}\overline{T}_k(x), \hbox{ for all }x \in [0,\lambda_{\max}],
\end{eqnarray}
where
\begin{align}\label{Eq:cheb_coeffs}
&\overline{T}_k(x):=T_k\left(\frac{x-\alpha}{\alpha}\right),  \nonumber \\
& \alpha:=\frac{\lambda_{\max}}{2}, \hbox{ and } \nonumber \\
& c_{j,k}:= \frac{2}{\pi}\int_{0}^{\pi}\cos(k\theta)~g_j\Bigl(\alpha \bigl(\cos(\theta) +1\bigr)\Bigr)~d\theta.
\end{align}
For $k \geq 2$, the shifted Chebyshev polynomials
satisfy
\begin{eqnarray*}
\overline{T}_k(x) = \frac{2}{\alpha}(x-\alpha)\overline{T}_{k-1}(x)
- \overline{T}_{k-2}(x).
\end{eqnarray*}
Thus, for any $f\in \R^N$, we have
\begin{equation} \label{eq:cheby-f-recurrence}
  \overline{T}_k(\L)f =
  \frac{2}{\alpha}({\L}-{\alpha}I)\left( \overline{T}_{k-1}(\L) f\right)
  - \overline{T}_{k-2}(\L)f,
\end{equation}
where $\overline{T}_k(\L) \in \Rbb^{N \times N}$ and the $n^{th}$ element of $\overline{T}_k(\L)f$ is given by
\begin{eqnarray}\label{Eq:T_bar_def}
\left(\overline{T}_k(\L)f\right)_n:=\sum_{\l=0}^{N-1} \overline{T}_k(\lambda_{\l})\hat{f}(\l)\chi_{\l}(n).
\end{eqnarray}

Now, to approximate the operator $\Phi$, we can approximate each 
multiplier $g_j(\cdot)$ by the first $M$ terms in its Chebyshev polynomial expansion \eqref{Eq:multiplier_expansion}.
Then,
for every $j \in \{1,2,\ldots,\eta\}$ and $n \in \{1,2,\ldots,N\}$, we have
\begin{align} \label{Eq:cheby_approx}
&\left(\tilde{\Phi} f\right)_{(j-1)N+n} \nonumber \\
&\quad\quad := \left(\frac{1}{2}c_{j,0}f+\sum_{k=1}^{M}c_{j,k}\overline{T}_k(\L)f\right)_n \\
&\quad \stackrel{(\ref{Eq:graph_IFT}), (\ref{Eq:T_bar_def})}= \sum_{\l=0}^{N-1}\left[\frac{1}{2}c_{j,0}+\sum_{k=1}^{M}c_{j,k}\overline{T}_k(\lambda_{\l})\right] \hat{f}(\l) \chi_{\l}(n) \nonumber \\
& \quad\quad \approx \sum_{\l=0}^{N-1}\left[\frac{1}{2}c_{j,0}+\sum_{k=1}^{\infty}c_{j,k}\overline{T}_k(\lambda_{\l})\right] \hat{f}(\l) \chi_{\l}(n)  \nonumber \\
&\quad\quad \stackrel{(\ref{Eq:multiplier_expansion})} = \sum_{\l=0}^{N-1}g_j(\lambda_{\l}) \hat{f}(\l) \chi_{\l}(n) \nonumber
\\
&\quad\quad \stackrel{(\ref{Eq:operator_def})} = \left(\Phi f\right)_{(j-1)N+n} . \nonumber
\end{align}

To recap, we propose to compute $\tilde{\Phi}f$ by first computing the Chebyshev coefficients $\{c_{j,k}\}_{j=1,2,\ldots,\eta;~k=1,2,\ldots,M}$ according to \eqref{Eq:cheb_coeffs}, and then computing the sum in \eqref{Eq:cheby_approx}. The computational benefit of the Chebyshev polynomial approximation arises in \eqref{Eq:cheby_approx}
from the fact the vector $\overline{T}_k(\L)f$ can be computed recursively from $\overline{T}_{k-1}(\L)f$ and $\overline{T}_{k-2}(\L)f$ according to \eqref{eq:cheby-f-recurrence}.
The computational cost of doing so is dominated by the matrix-vector multiplication of the graph Laplacian $\L$, which is proportional to the number of edges, $|E|$ \cite{LTS-ARTICLE-2009-053}. Therefore, if the underlying communication graph is sparse (i.e., $|E|$ scales linearly with the network size $N$), it is far more computationally efficient to compute $\tilde{\Phi}f$ than $\Phi f$. Finally, we note that in practice, setting the Chebyshev approximation order $M$ to around 20 results in $\tilde{\Phi}$ approximating $\Phi$ very closely in all of the applications we have examined.

\section{Distributed Computation} \label{Se:distribution}
In the previous section, we showed that the Chebyshev polynomial approximation to a union of graph Fourier multipliers provides computational efficiency gains, even in a centralized computation setting. In this section, we discuss the second benefit of the Chebyshev polynomial approximation: it is easily distributable.

\subsection{Distributed Computation of $\tilde{\Phi}f$} \label{Se:forward}
We consider the following scenario. There is a network of $N$ nodes, and each node $n$ begins with the following knowledge:
\begin{itemize}
\item $f(n)$, the $n^{th}$ component of the signal $f$
\item The identity of its neighbors, and the weights of the graph edges connecting itself to each of its neighbors
\item The first $M$ Chebyshev coefficients, $c_{j,k}$, for $j \in \{1,2,\ldots,\eta\}$ and $k \in \{0,1,2,\ldots,M\}$. These can either be computed centrally according to \eqref{Eq:cheb_coeffs} and then transmitted throughout the network, or each node can begin with knowledge of the
multipliers, $\{g_j(\cdot)\}_{j=1,2,\ldots,\eta}$, and precompute the 
Chebyshev coefficients according to \eqref{Eq:cheb_coeffs}
\item An upper bound on $\lambda_{\max}$, the largest eigenvalue of the graph Laplacian. This bound need not be 
tight, so we can precompute a bound such as $\lambda_{\max}\leq \max\{d(m)+d(n); m \sim n \}$, where $d(n)$ is the degree of node $n$ \cite{anderson_morley} 
\end{itemize}

The task is for each network node $n$ to compute
\begin{eqnarray}\label{Eq:ntarget}
\left\{\left(\tilde{\Phi} f\right)_{(j-1)N+n}\right\}_{j=1,2,\ldots,\eta}
\end{eqnarray}
by iteratively exchanging messages with its local neighbors in the network and performing some computations.

As a result of \eqref{Eq:cheby_approx}, for node $n$ to compute the desired sequence in \eqref{Eq:ntarget}, it suffices to learn $\left\{\left(\overline{T}_k(\L)f\right)_n\right\}_{k=1,2,\ldots, M}$. 
Note that $\left(\overline{T}_1({\L})f\right)_n=\left(\frac{1}{\alpha}({\L}-{\alpha}I)f\right)_n$ and $\L_{n,m}=0$ for all nodes $m$ that are not neighbors of node $n$. Thus, to compute $\left(\overline{T}_1({\L})f\right)_n$, sensor node $n$ just needs to receive $f(m)$ from all neighbors $m$. So once all nodes send their component of the signal to their neighbors, they are able to compute their respective components of $\overline{T}_1({\L})f$. In the next step, each node $n$ sends the newly computed quantity $\left(\overline{T}_1({\L})f\right)_n$ to all of its neighbors, enabling the distributed computation of $\overline{T}_2({\L})f$ according to \eqref{eq:cheby-f-recurrence}. The iterative process of computation and local communication continues for $M$ rounds until each node $n$ has computed the required sequence $\left\{\left(\overline{T}_k(\L)f\right)_n\right\}_{k=1,2,\ldots, M}$. In all,
$2M\card{{E}}$ messages of length 1 are required for every node $n$ to compute its sequence of coefficients in \eqref{Eq:ntarget} in a distributed fashion. This distributed computation process is summarized in Algorithm 1. 

\begin{algorithm}[t] \label{alg1}
\caption{Distributed Computation of $\tilde{\Phi}f$}
      Inputs at node $n$: $f_n$, $\L_{n,m}~\forall m$, $\left\{c_{k,j}\right\}_{j=1,2,\ldots,\eta;~k=0,1,\ldots,M}$, \\
      and $\lambda_{\max}$ \\
   Outputs at node $n$: $\left\{\left(\tilde{\Phi} f\right)_{(j-1)N+n}\right\}_{j=1,2,\ldots,\eta}$ \\ 
       \begin{algorithmic}[1]
   \STATE Set $\left(\overline{T}_0(\L)f\right)_n = f_n$
   \STATE Transmit $f_n$ to all neighbors ${\cal N}_n:=\{m:\L_{n,m} < 0\}$
   \STATE Receive $f_m$ from all neighbors ${\cal N}_n$
   \STATE Compute and store
   \begin{align*}
   \left(\overline{T}_1(\L)f\right)_n =
  \sum\limits_{m \in {\cal N}_n \cup n}\frac{1}{\alpha}\L_{n,m}f_m  -f_n
  \end{align*}
   \FOR{$k=2,\ldots,M$}
   \STATE Transmit $\left(\overline{T}_{k-1}(\L)f\right)_n$ to all neighbors ${\cal N}_n$
    \STATE Receive $\left(\overline{T}_{k-1}(\L)f\right)_m$ from all neighbors ${\cal N}_n$
   \STATE Compute and store
   \begin{align*}
   \left(\overline{T}_k(\L)f\right)_n =&
 \sum\limits_{m \in {\cal N}_n \cup n} \frac{2}{\alpha}\L_{n,m}\left(\overline{T}_{k-1}(\L)f\right)_m  \\
 &-2\left(\overline{T}_{k-1}(\L)f\right)_n
  - \left(\overline{T}_{k-2}(\L)f\right)_n
  \end{align*}
   \ENDFOR
   \FOR{$j \in \{1,2,\ldots,\eta\}$}
   \STATE Output
   \begin{align*}
   \left(\tilde{\Phi} f\right)_{(j-1)N+n} = \frac{1}{2}c_{j,0}f_n+\sum\limits_{k=1}^M c_{j,k} \left(\overline{T}_k(\L)f\right)_n
   \end{align*}
   \ENDFOR
   \end{algorithmic}
   \end{algorithm}


An important point to emphasize again is that although the operator $\Phi$ and its approximation $\tilde{\Phi}$ are defined through the eigenvectors of the graph Laplacian, the Chebyshev polynomial approximation helps the sensor nodes 
apply the operator to the signal without explicitly computing (individually or collectively) the eigenvalues or eigenvectors of the Laplacian, other than the upper bound on its spectrum. Rather, they initially communicate their component of the signal to their neighbors, and then communicate simple weighted combinations of the messages received in the previous stage in subsequent iterations. 
In this way, information about each component of the signal $f$ diffuses through the network
without direct communication between non-neighboring nodes.

\subsection{Distributed Computation of $\tilde{\Phi}^*a$} \label{Se:adj_d}
The application of the adjoint $\tilde{\Phi}^*$ of the Chebyshev polynomial approximate operator $\tilde{\Phi}$ can also be computed in a distributed manner. Let
\begin{eqnarray*}
a  = \left[ a_1; a_2; \ldots; a_{\eta} \right] \in \Rbb^{\eta N},
\end{eqnarray*}
where $a_j \in \Rbb^N$. Then it is straightforward to show that
\begin{align}\label{Eq:adjoint_split}
\left(\tilde{\Phi}^*a\right)_n=\sum_{j=1}^{\eta}\left(\frac{1}{2}c_{j,0}a_j+\sum_{k=1}^{M}c_{j,k}\overline{T}_k(\L)a_j \right)_n.
\end{align}
We assume each node $n$ starts with knowledge of $a_j(n)$ for all $j \in \{1,2,\ldots,\eta\}$. For each $j \in \{1,2,\ldots,\eta\}$, the distributed computation of the corresponding term on the right-hand side of \eqref{Eq:adjoint_split} is done in an analogous manner to the distributed computation of $\tilde{\Phi}f$ discussed above. Since this has to be done for each $j$, $2M|E|$ messages, each a vector of length $\eta$, are required for every node $n$ to compute $\left(\tilde{\Phi}^*a\right)_n$.

\subsection{Distributed Computation of $\tilde{\Phi}^*\tilde{\Phi}f$}\label{Se:wwstar} 
Using the property of Chebyshev polynomials that
\begin{eqnarray*}
T_k(x)T_{k^{\prime}}(x)=\frac{1}{2}\left[T_{k+k^{\prime}}(x)+T_{|k-k^{\prime}|}(x) \right],
\end{eqnarray*}
we can write (see \cite{LTS-ARTICLE-2009-053} for a similar calculation)
\begin{eqnarray*}
\left(\tilde{\Phi}^*\tilde{\Phi}f\right)_n=\left(\frac{1}{2}d_0 f + \sum_{k=1}^{2M} d_k \overline{T}_k(\L)f \right)_n.
\end{eqnarray*}
Therefore, with each node $n$ starting with $f(n)$ as in Section \ref{Se:forward}, the nodes can compute $\tilde{\Phi}^*\tilde{\Phi}f$ in a distributed manner using $4M|E|$ messages of length 1, with each node $n$ finishing with knowledge of $\left(\tilde{\Phi}^*\tilde{\Phi}f\right)_n$.





\section{Application Examples}\label{Se:applications}
In this section, we provide more detailed explanations of how the Chebyshev polynomial approximation of graph Fourier multipliers can be used in the context of specific distributed signal processing tasks.
\subsection{Distributed Smoothing}
Perhaps the simplest example application is distributed smoothing with the heat kernel as the graph Fourier multiplier. One way to smooth a signal $y\in \Rbb^N$ is to compute $H_t y$, where, for a fixed $t$,
$(H_ty)(n):=\sum_{\l=0}^{N-1} e^{-t\lambda_{\l}}\hat{y}(\l)\chi_{\l}(n)$.
$H_t$ clearly satisfies our definition of a graph Fourier multiplier operator (with $\eta=1$).
In the context of a centralized image smoothing application, \cite{hancock} discusses in detail the \emph{heat kernel}, $H_t$, and its relationship to classical Gaussian filtering. Similar to the example at the end of Section \ref{Se:gfmo}, the main idea is 
that the multiplier $e^{-t\lambda_{\l}}$ acts as a low-pass filter that attenuates the higher frequency (less smooth) components of 
$y$.

Now, to perform distributed smoothing, we just need to compute $\tilde{H}_t y$ in a distributed manner according to Algorithm 1, where $\tilde{H}_t$ is the shifted Chebyshev polynomial approximation to the graph Fourier multiplier operator $H_t$.



\subsection{Distributed Regularization} \label{Se:reg}
Regularization is a common signal processing technique to solve ill-posed inverse problems using \emph{a priori} information about a target signal to recover it accurately. Here we use regularization to solve the distributed denoising task discussed in Section \ref{sec:intro}, starting with a noisy signal $y\in\R^N$ defined on a graph of $N$ sensors. The prior belief we want to
enforce is that the target signal is smooth with respect to the underlying graph topology. The class of regularization terms we consider is $f^{\transpose}\L^r f$ for $r \geq 1$, and the resulting regularization problem has the form
\begin{eqnarray}\label{Eq:reg_prob}
\argmin_f \frac{\tau}{2}\norm{f-y}_2^2+f^{\transpose}\L^r f.
\end{eqnarray}
To see intuitively why incorporating such a regularization term into the objective function encourages smooth signals (with $r=1$ as an example), note that $f^{\transpose}\L f=0$ if and only if $f$ is constant across all vertices, and, more generally
\begin{eqnarray*}
f^{\transpose}\L f=\frac{1}{2}\sum_{n \in V}\sum_{m \sim n}A_{m,n}\left[f(m)-f(n)\right]^2, 
\end{eqnarray*}
so $f^{\transpose}\L f$ is small when the signal $f$ has similar values at neighboring vertices with large weights (i.e., it is smooth).

We now show how our novel method is useful in solving this distributed regularization problem.
\begin{proposition} \label{Prop:reg}
The solution to \eqref{Eq:reg_prob} is given by $Ry$, where $R$ is a graph Fourier multiplier operator of the form \eqref{Eq:gfm_def}, with multiplier $g(\lambda_{\l})=\frac{\tau}{\tau+2\lambda_{\l}^r}$ .\footnote{This filter $g(\lambda_{\l})$ is the graph analog of a first-order Bessel filter 
from classical signal processing of functions on the real line.}
\end{proposition}
\begin{proof}
The objective function in \eqref{Eq:reg_prob} is convex in $f$. 
Differentiating it with respect to $f$, 
any solution $f_*$ to
\begin{eqnarray}\label{Eq:opt_eq}
\L^r f_* + \frac{\tau}{2}(f_*-y)=0
\end{eqnarray}
is a solution to \eqref{Eq:reg_prob}.\footnote{In the case $r=1$, the optimality equation \eqref{Eq:opt_eq} corresponds to the optimality equation in \cite[Section III-A]{elmoataz} with $p=2$ in that paper.} Taking the graph Fourier transform of \eqref{Eq:opt_eq} yields
\begin{eqnarray}\label{Eq:opt_eq_four}
&\widehat{\L^r f_*}(\l) + \frac{\tau}{2}\left(\widehat{f_*}(\l)-\hat{y}(\l)\right)=0, \\
&\hspace{1.3in}\forall \l \in \{0,1,\ldots,N-1\}. \nonumber
\end{eqnarray}
From the real, symmetric nature of $\L$ and the definition of the Laplacian eigenvectors ($\L \chi_{\l} = \lambda_{\l} \chi_{\l}$), we have:
\begin{eqnarray} \label{Eq:lap}
\widehat{\L^r f_*}(\l) = \chi_{\l}^* \L^r f_* = \left(\L^r \chi_{\l}\right)^*f_* = \lambda_{\l}^r\chi_{\l}^* f_* = \lambda_{\l}^r \widehat{f_*}(\l).
\end{eqnarray}
Substituting \eqref{Eq:lap} into \eqref{Eq:opt_eq_four} and rearranging, we have
\begin{eqnarray} \label{Eq:rearr}
\widehat{f_*}(\l)=\frac{\tau}{\tau+2\lambda_{\l}^r}\hat{y}(\l),~\forall \l \in \{0,1,\ldots,N-1\}.
\end{eqnarray}
Finally, taking the inverse graph Fourier transform of \eqref{Eq:rearr}, we have
\begin{eqnarray}\label{Eq:reg_multiplier}
f_*(n)=\sum_{\l=0}^{N-1}\widehat{f_*}(\l)\chi_{\l}(n)=\sum_{\l=0}^{N-1}\left[\frac{\tau}{\tau+2\lambda_{\l}^r}\right]\hat{y}(\l)\chi_{\l}(n),
\\
\forall n \in \{1,2,\ldots,N\}. \nonumber
\end{eqnarray}
\end{proof}

So, one way to do distributed denoising is 
to compute $\tilde{R}y$, the Chebyshev polynomial approximation of $Ry$, in a distributed manner via Algorithm 1. 
We show this now with a numerical example.
We place 500 sensors randomly in the $[0,1] \times [0,1]$ square. We then construct a weighted graph according to the thresholded Gaussian kernel weighting \eqref{Eq:gkw} with $\sigma=0.074$ and $\kappa=0.600$, so that two sensor nodes are connected if 
their physical separation is less than 
0.075. We create a smooth 500-dimensional signal with the $n^{th}$ component given by
$f_n^0 = n_x^2+n_y^2-1$,
where $n_x$ and $n_y$ are 
node $n$'s $x$ and $y$ coordinates in $[0,1]\times[0,1]$. One instance of such a network and signal $f^0$ are shown in Figure \ref{Fig:network}, and the eigenvectors of the graph Laplacian 
are shown in Figure \ref{Fig:eigenvectors}.

Next, we corrupt each component of the signal $f^0$ with uncorrelated additive Gaussian noise with mean zero and standard deviation 0.5. Then we apply the graph Fourier multiplier operator $\tilde{R}$, the Chebyshev polynomial approximation to $R$ from Proposition \ref{Prop:reg}, with $\tau=r=1$. The multiplier and its Chebyshev polynomial approximations are shown in Figure \ref{Fig:filter}, and the denoised signal $\tilde{R}y$ is shown in Figure \ref{Fig:denoising}. We repeated this entire experiment 1000 times, with a new random graph and random noise each time, and the average mean square error for the denoised signals was 0.013, as compared to 0.250 average mean square error for the noisy signals.
\begin{figure}
\centering{
\includegraphics[width=2.55in]{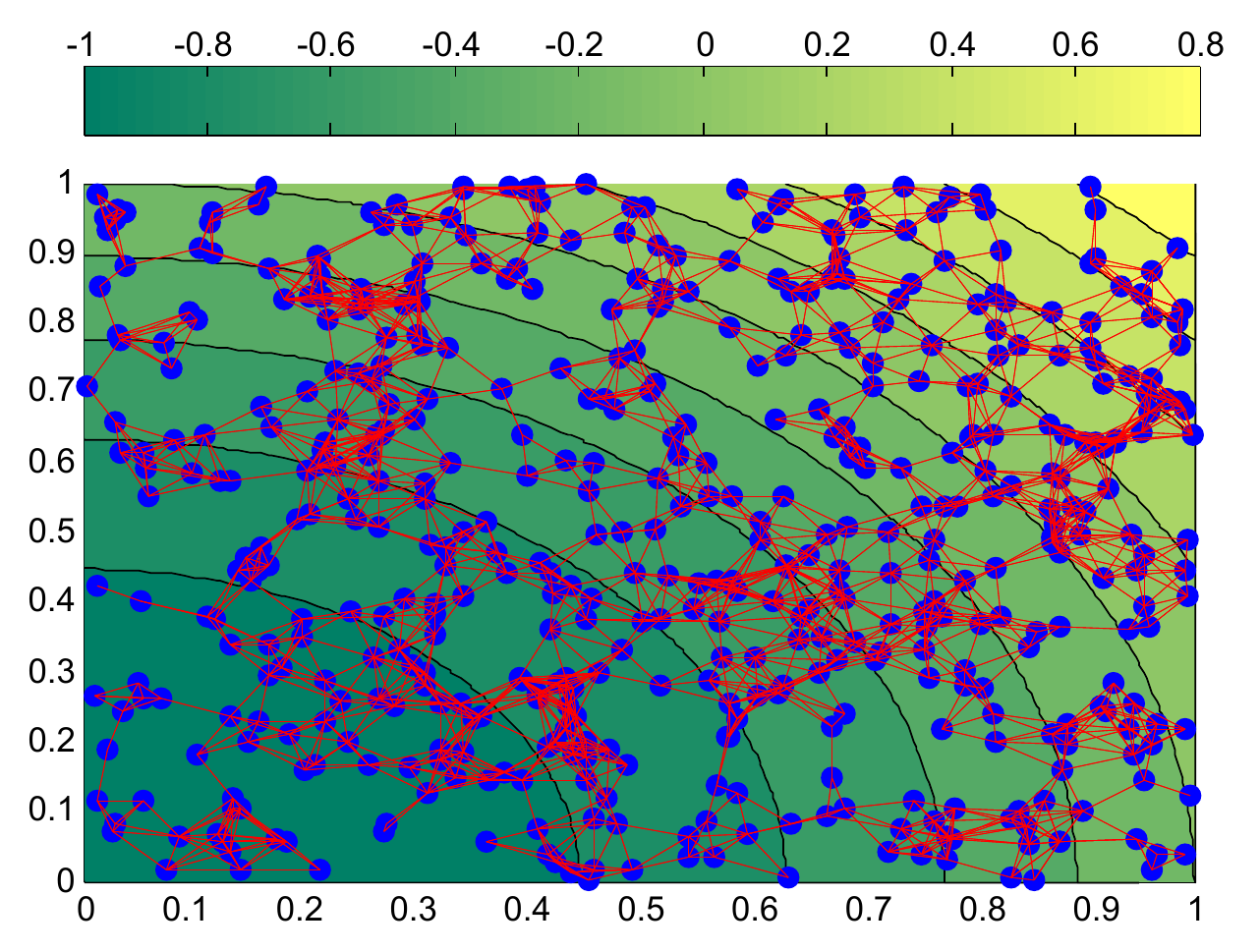}
}\caption{A network of 500 sensors placed randomly in the $[0,1] \times [0,1]$ square. The background colors represent the values of the smooth signal $f^0$. } \label{Fig:network} 
\end{figure}
\begin{figure}[tb]
\centering
\begin{minipage}[b]{0.43\linewidth}
   \centering
   \centerline{\small{$\chi_0$}}
   \centerline{\includegraphics[width=1\linewidth]{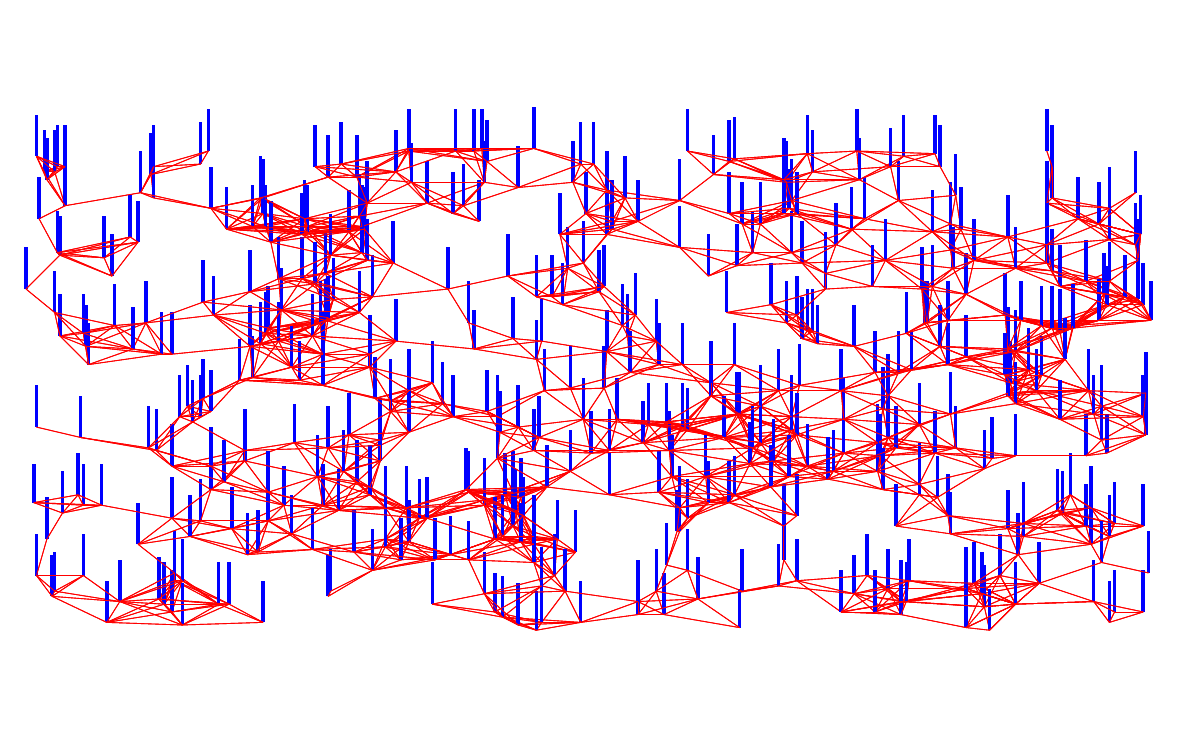}}
\centerline{\small{(a)}}
\end{minipage}
\hfill
\begin{minipage}[b]{0.43\linewidth}
   \centering
   \centerline{\small{$\chi_1$}}
   \centerline{\includegraphics[width=1\linewidth]{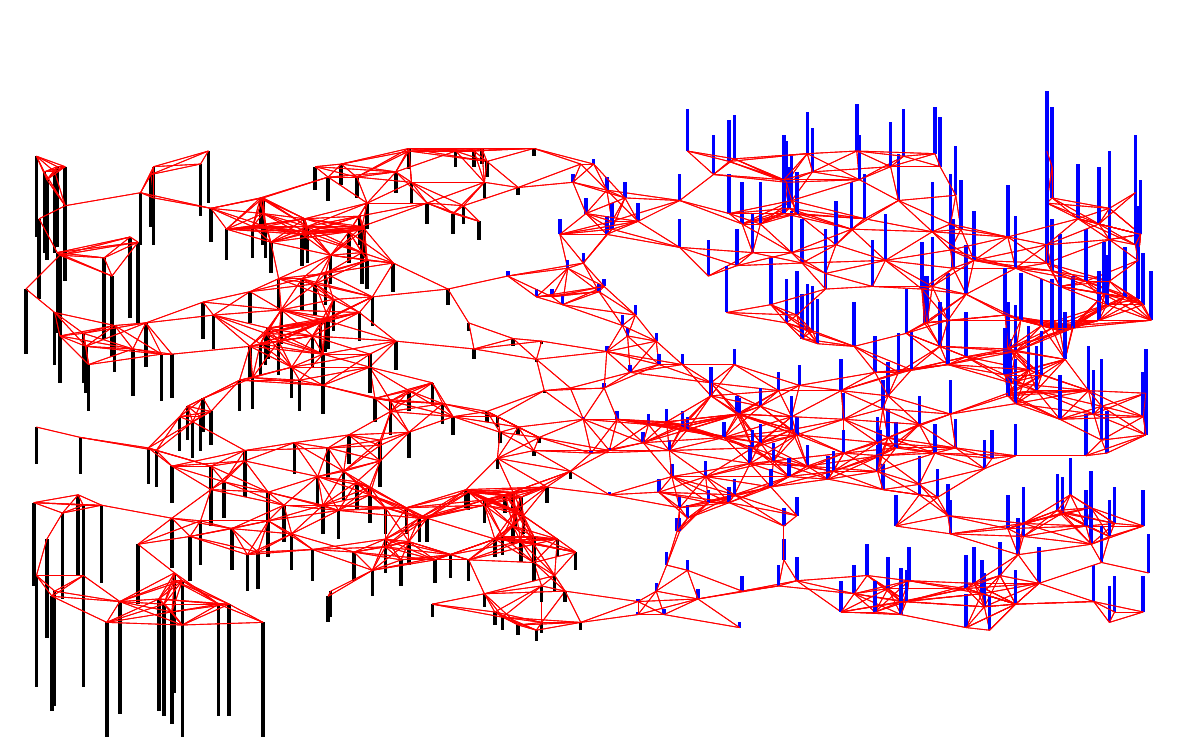}}
\centerline{\small{(b)}}
\end{minipage}\\
\vspace{0.4cm}
\begin{minipage}[b]{0.43\linewidth}
   \centering
   \centerline{\small{$\chi_2$}}
   \centerline{\includegraphics[width=1\linewidth]{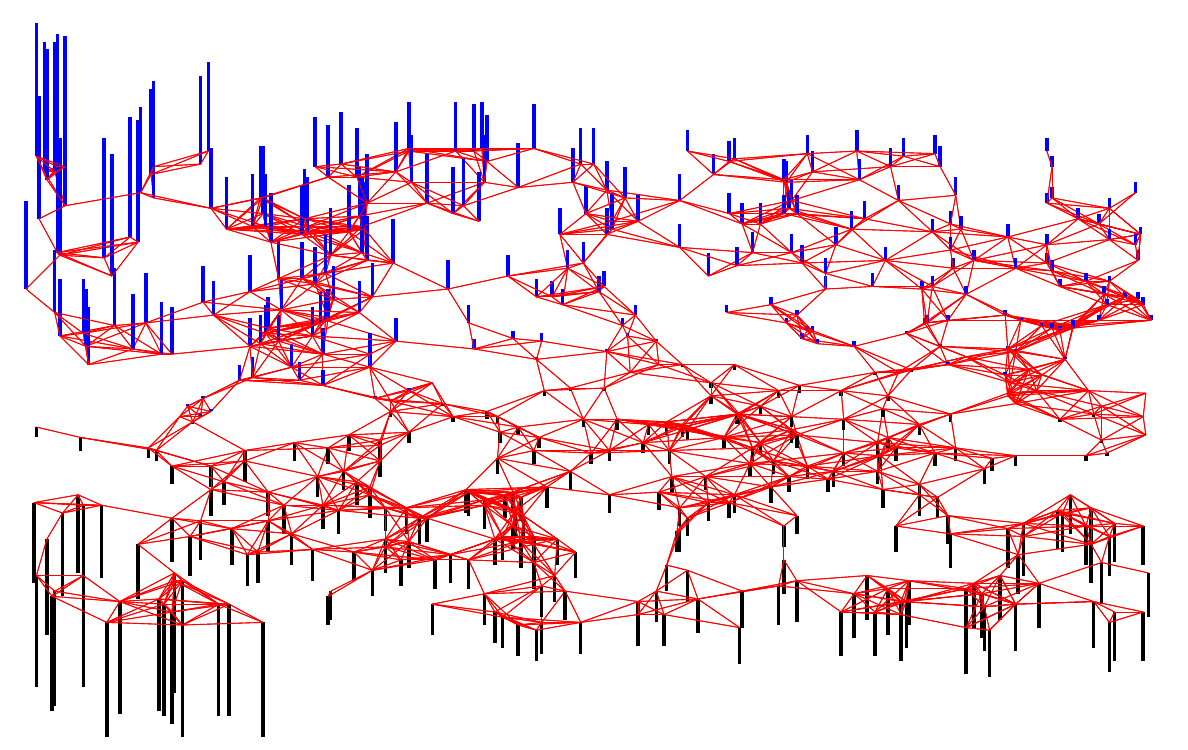}}
\centerline{\small{(c)}}
\end{minipage}
\hfill
\begin{minipage}[b]{0.43\linewidth}
   \centering
   \centerline{\small{$\chi_{50}$}}
   \centerline{\includegraphics[width=1\linewidth]{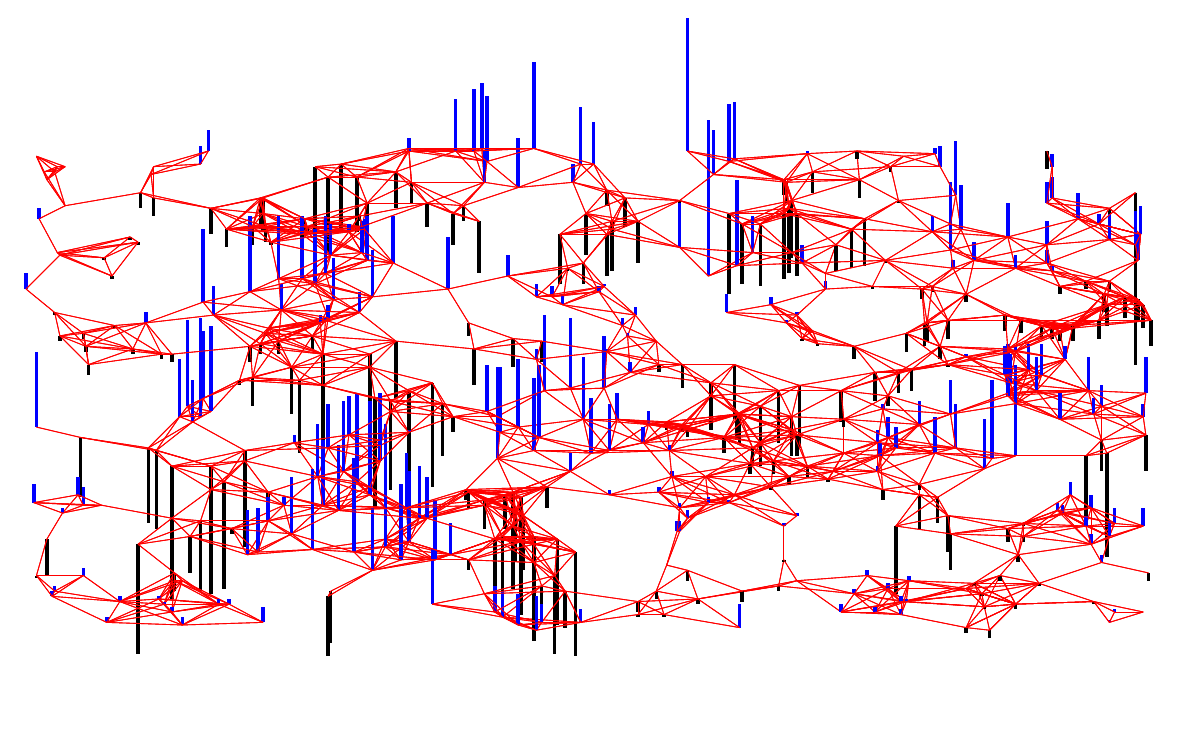}}
\centerline{\small{(d)}}
\end{minipage}\\
\caption {Some eigenvectors of the Laplacian of the graph shown in Figure \ref{Fig:network}. The blue bars represent positive values and the black bars negative values. (a) $\chi_0$, the constant eigenvector associated with $\lambda_{0}=0$. 
(b) $\chi_1$, the \emph{Fiedler vector} associated with the lowest strictly positive eigenvalue, nicely separates the graph into two components. (c) $\chi_2$ is also a smooth eigenvector. (d) $\chi_{50}$ is far less smooth
with some large differences across neighboring nodes.}
  \label{Fig:eigenvectors}
\end{figure}
\begin{figure}
\centering{
\includegraphics[width=2in]{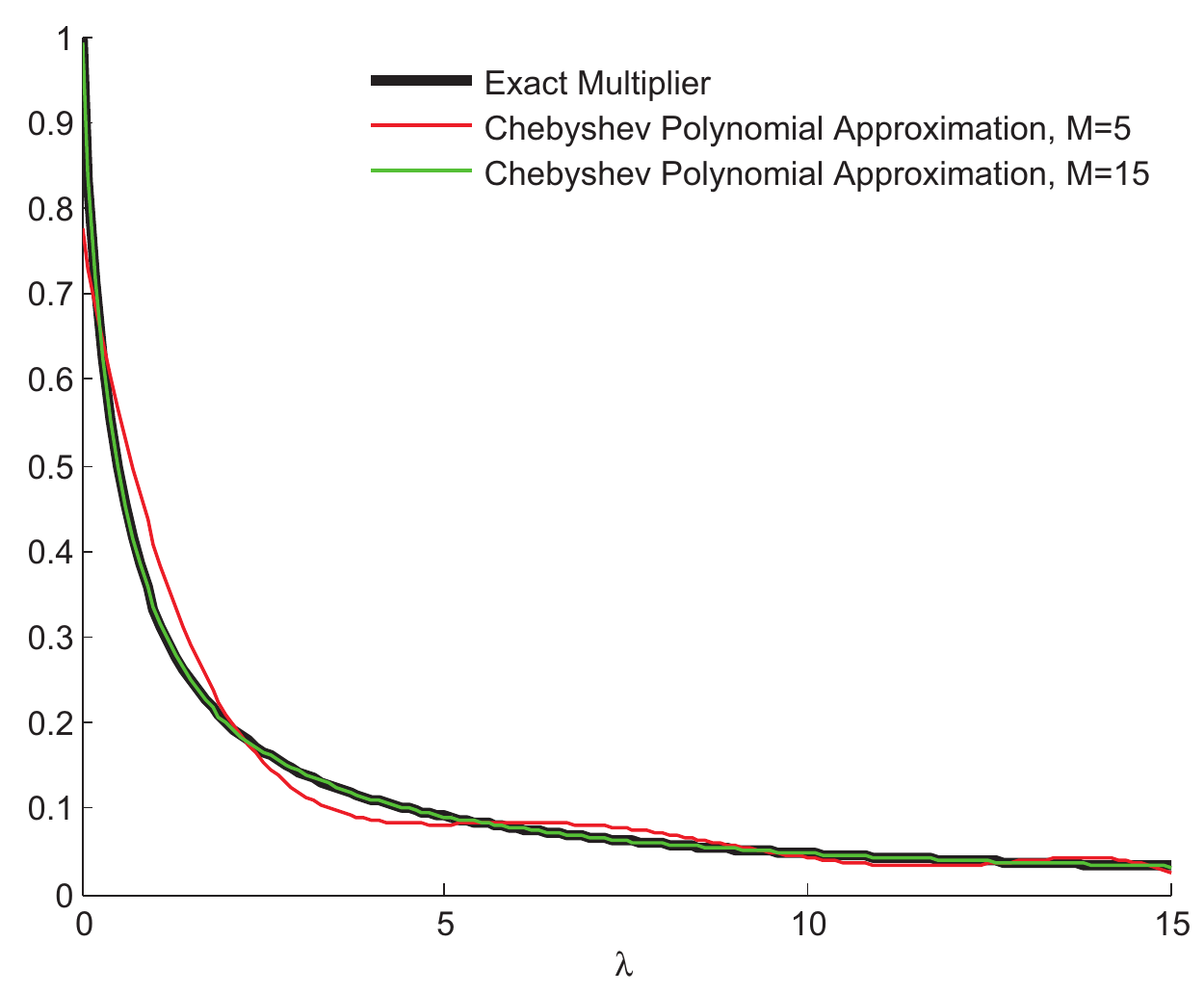}
}\caption{The regularizing multiplier $\frac{\tau}{\tau+2\lambda_{\l}^r}$ associated with the graph Fourier multiplier operator $R$ from Proposition \ref{Prop:reg}. Here, $r=\tau=1$. Shifted Chebyshev polynomial approximations to the multiplier are shown for different values of the approximation order $M$.} \label{Fig:filter}
\end{figure}
\begin{figure}[tb]
\centering
\begin{minipage}[b]{0.43\linewidth}
   \centering
   \centerline{\small{Original Signal}}
   \centerline{\includegraphics[width=1\linewidth]{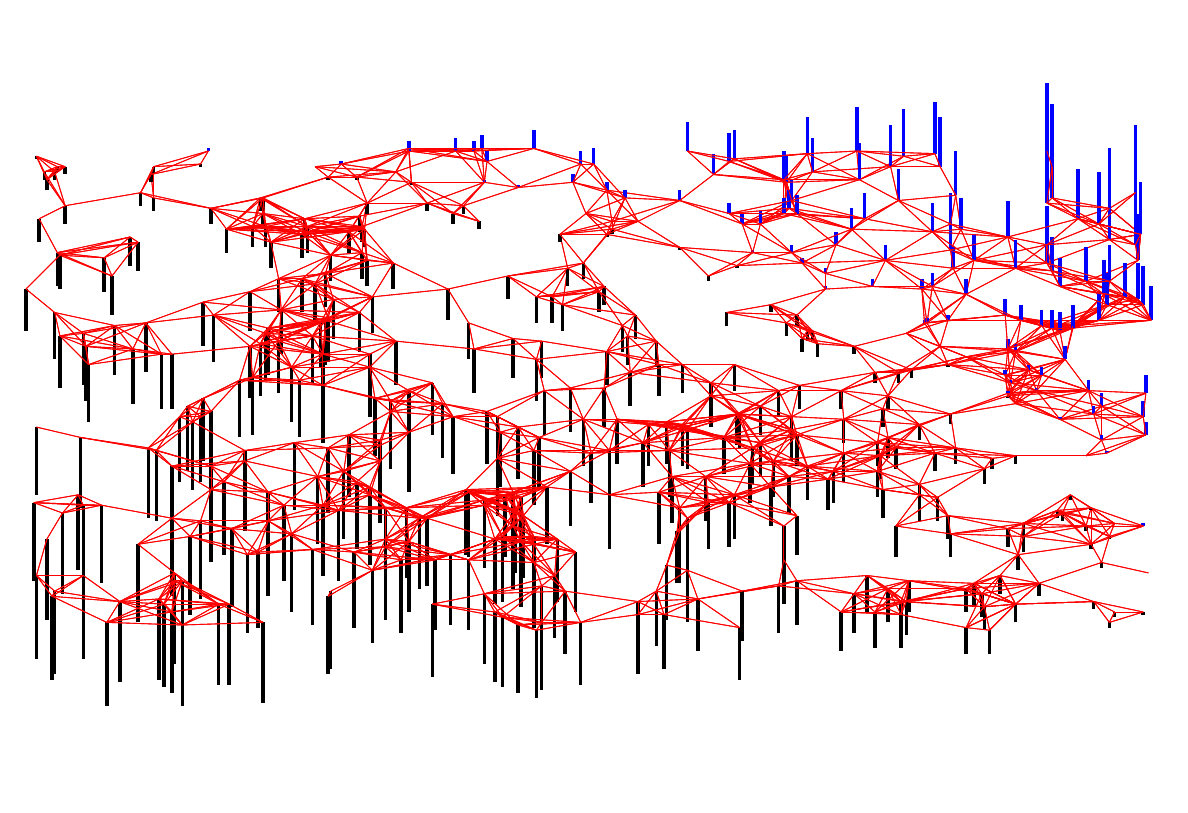}}
\centerline{\small{(a)}}
\end{minipage}
\hfill
\begin{minipage}[b]{0.43\linewidth}
   \centering
   \centerline{\small{Noise}}
   \centerline{\includegraphics[width=1\linewidth]{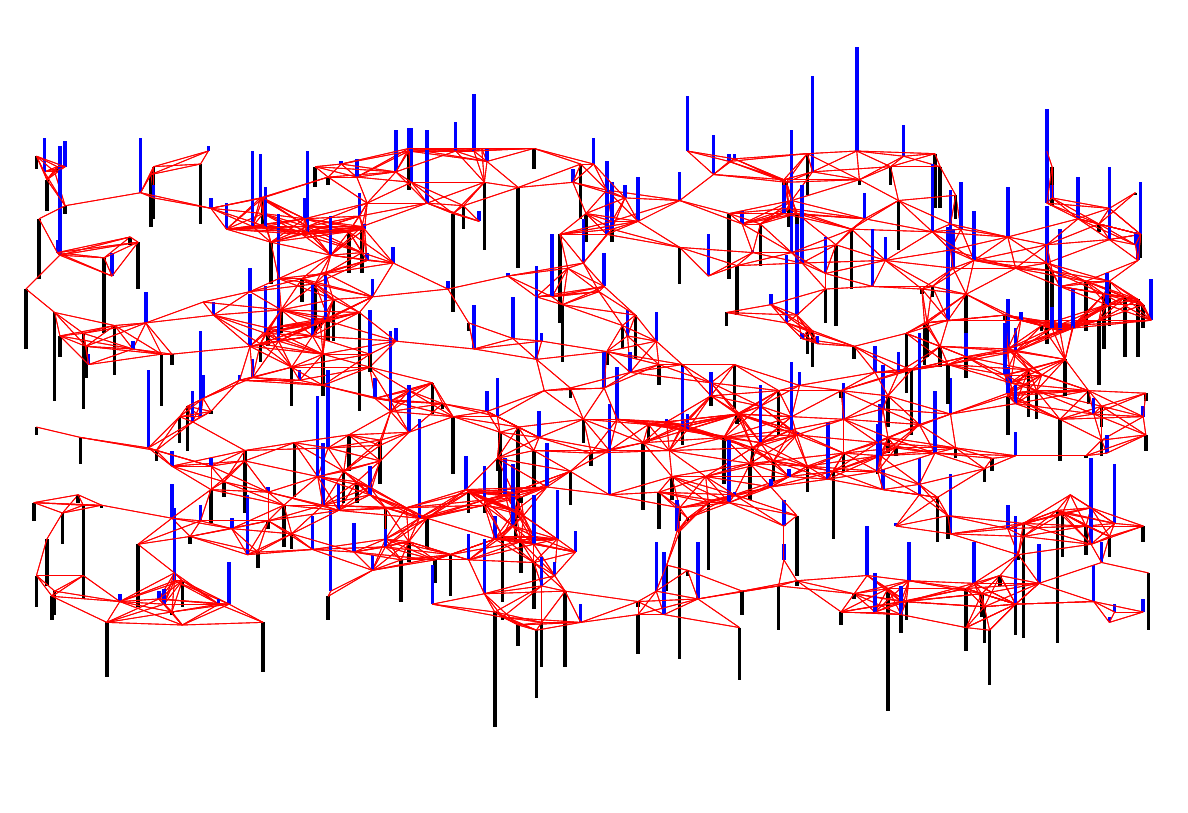}}
\centerline{\small{(b)}}
\end{minipage}\\
\vspace{0.4cm}
\begin{minipage}[b]{0.43\linewidth}
   \centering
   \centerline{\small{Noisy Signal}}
   \centerline{\includegraphics[width=1\linewidth]{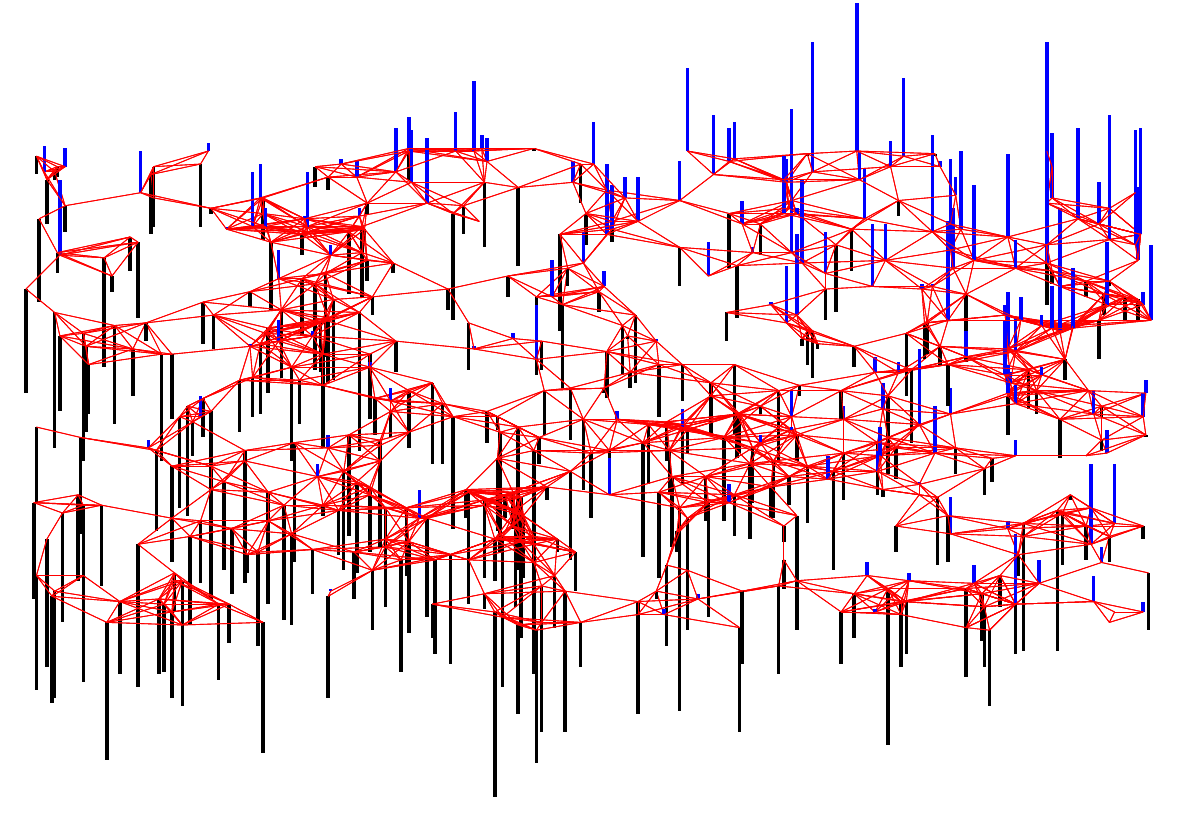}}
\centerline{\small{(c)}}
\end{minipage}
\hfill
\begin{minipage}[b]{0.43\linewidth}
   \centering
   \centerline{\small{Denoised Signal}}
   \centerline{\includegraphics[width=1\linewidth]{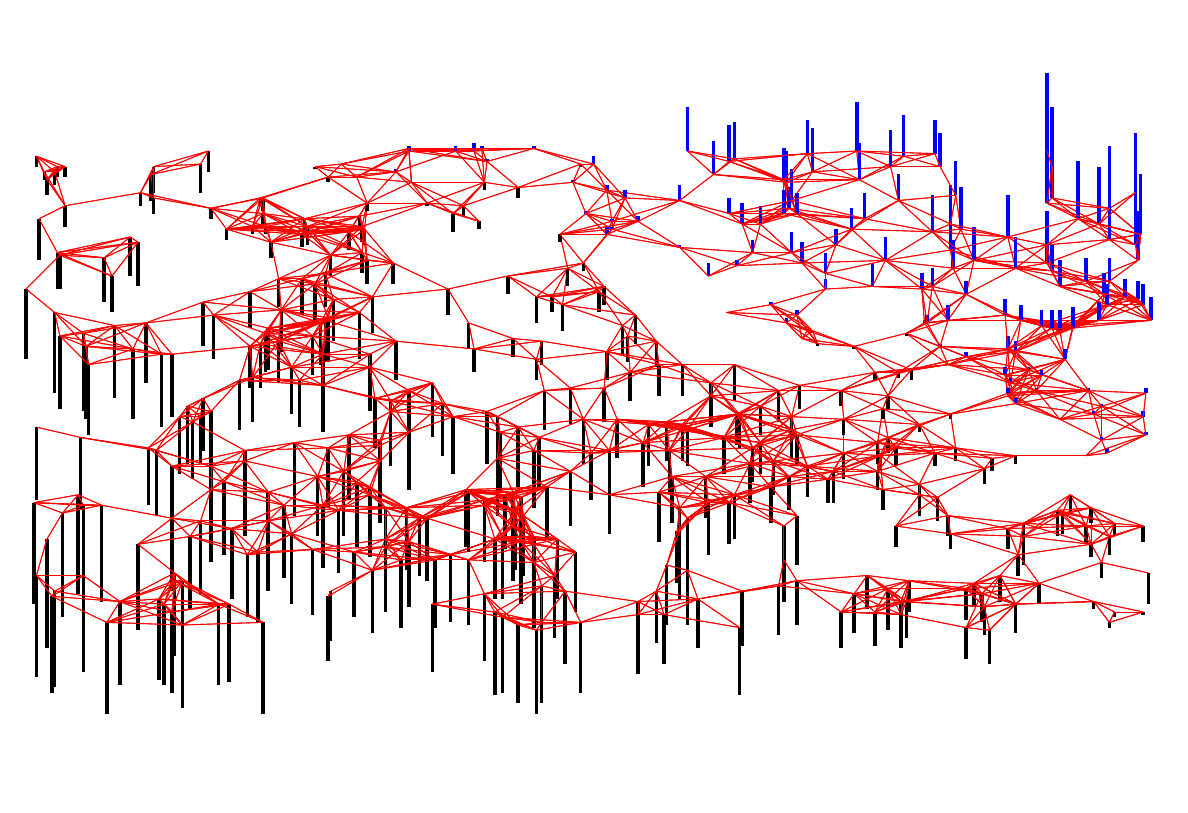}}
\centerline{\small{(d)}}
\end{minipage}\\
\caption {A denoising example on the graph shown in Figure \ref{Fig:network}, using the regularizing multiplier shown in Figure \ref{Fig:filter}. (a) The original signal $n_x^2+n_y^2-1$, where $n_x$ and $n_y$ are the $x$ and $y$ coordinates of sensor node $n$. (b) The additive Gaussian noise. (c) The noisy signal $y$. (d) The denoised signal $\tilde{R}y$.}
  \label{Fig:denoising}
\end{figure}

We conclude this section by returning 
to the distributed binary classification task discussed in the introduction. In \cite{belkin_matveeva}, Belkin \emph{et al.} show that the regularizer $f^{\transpose}\L^r f$ also works well in graph-based semi-supervised learning. One approach to distributed binary classification is to let $y_n$ be the labels (-1 or 1) of those nodes who know their labels, and 0 otherwise. Then the nodes compute $\tilde{R}y$ in a distributed manner via Algorithm 1, and each node $n$ sets it label to 1 if
$(\tilde{R}y)_n \geq 0$ and -1 otherwise. We believe our approach to distributedly applying graph Fourier multipliers 
can also be used
for more general 
learning problems, but we leave this for future work.


\subsection{Distributed Wavelet Denoising}\label{Se:denoising}
In this section, we consider an alternate method of distributed denoising that may be better suited to signals that are piecewise smooth on the graph, but not necessarily globally smooth. The setup is the same as in Section \ref{Se:reg}, with a noisy signal $y \in \R^N$,
and each sensor $n$ observing $y_n$. Instead of starting with a prior that the signal is globally smooth, we start with a prior belief that the signal is sparse in the spectral graph wavelet domain \cite{LTS-ARTICLE-2009-053}.
The spectral graph wavelet transform, $W$, defined in \cite{LTS-ARTICLE-2009-053}, is precisely of the form of $\Phi$ in \eqref{Eq:operator_def}. Namely, it is composed of one 
multiplier, $h(\cdot)$, that acts as a low-pass filter to stably represent the signal's low frequency content, and $J$ wavelet operators, defined by $g_j(\lambda_{\l})=g(t_j \lambda_{\l})$, where $\{t_j\}_{j=1,2,\ldots,J}$ is a set of scales and $g(\cdot)$ is the wavelet multiplier 
that acts as a band-pass filter.

The most common way to incorporate a sparse prior in a centralized setting is to regularize via a weighted version of the \emph{least absolute shrinkage and selection operator (lasso)} \cite{lasso}, also called \emph{basis pursuit denoising} \cite{basispursuit}:
\begin{eqnarray} \label{Eq:lasso}
\argmin_{a}~\frac{1}{2}\norm{y-W^*{a}}_2^2+\norm{{a}}_{1,\mu}~,
\end{eqnarray}
where
$\norm{{a}}_{1,\mu}:=\sum_{i=1}^{N(J+1)} \mu_i \left|{a}_i\right|$.
The optimization problem in \eqref{Eq:lasso} can be solved for example by iterative soft thresholding \cite{DDD}. The initial estimate of the wavelet coefficients ${a}^{(0)}$ is arbitrary, and at each iteration of the soft thresholding algorithm, the update of the estimated wavelet coefficients is given by
\begin{eqnarray}\label{Eq:ista_update}
{a}_i^{(k)}=\Scal_{{\mu}_i {\tau}}\left(\Bigl({a}^{(k-1)}+{\tau}W\left[y-W^*{a}^{(k-1)}\right]\Bigr)_i\right),  \nonumber \\
i=1,2,\ldots,N(J+1);~k=1,2,\ldots
\end{eqnarray}
where ${\tau}$ is the step size and $\Scal_{{\mu}_i{\tau}}$ is the shrinkage or soft thresholding operator
\begin{eqnarray*}
\Scal_{{\mu}_i{\tau}}(z):=\left\{
\begin{array}{ll}
0&,\mbox{ if } \mid z \mid \leq {\mu}_i{\tau} \\
z-\mbox{sgn}(z){\mu}_i{\tau}&, \mbox{ o.w.}
\end{array}
\right. .
\end{eqnarray*}
The iterative soft thresholding algorithm converges to ${{a}}_*$, the minimizer of \eqref{Eq:lasso}, if ${\tau} < \frac{2}{{\norm{W^*}^2}}$ \cite{combettes}. The final denoised estimate of the signal is then given by $W^*{{a}}_*$.

We now turn to the issue of how to implement the above algorithm in a distributed fashion by sending messages between neighbors in the network. 
One option would be to use the distributed lasso algorithm of \cite{dlasso}, which is a special case of the Alternating Direction Method of Multipliers \cite[p.~253]{par_dist_opt_book}. In every iteration of that algorithm, each node transmits its current estimate of \emph{all} the wavelet coefficients to its local neighbors. 
With a transform the size of the spectral graph wavelet transform, 
this requires $2\card{{E}}$ total messages 
at every iteration, with each message being a vector of length  $N(J+1)$.
A method where the amount of communicated information does not grow with $N$ (beyond the number of edges, $\card{{E}}$) would be highly preferable.

The 
Chebyshev polynomial approximation of the spectral graph wavelet transform
allows us to accomplish this goal.
Our approach is to approximate $W$ by $\tilde{W}$, and use the distributed implementation of the approximate wavelet transform and its adjoint to
perform iterative soft thresholding.
In the first soft thresholding iteration, each node $n$ must learn $(\tilde{W}y)_{(j-1)N+n}$ at all scales $j$, via Algorithm 1. These coefficients are then stored for future iterations. In the $k^{th}$ iteration, each node $n$ must learn the $J+1$ coefficients of $\tilde{W}\tilde{W}^*{a}^{(k-1)}$ centered at $n$, by sequentially applying the operators $\tilde{W}^*$ and $\tilde{W}$ in a distributed manner via the methods of Sections \ref{Se:adj_d} and \ref{Se:forward}, respectively. Finally, when a stopping criterion for the soft thresholding is satisfied,
the adjoint operator $\tilde{W}^*$ is applied again in a distributed manner to the resulting coefficients $\tilde{{{a}}}_*$, 
and node $n$'s denoised estimate of its signal is $\left(\tilde{W}^*\tilde{{{a}}}_*\right)_n$.

We now examine the communication requirements of this approach. Recall from Section \ref{Se:distribution} that
$2M\card{{E}}$ messages of length 1 are required to compute $\tilde{W}y$ in a distributed fashion.
Distributed computation 
of $\tilde{W}\tilde{W}^*{a}^{(k-1)}$, the other term needed in the iterative thresholding update \eqref{Eq:ista_update}, requires $2M\card{{E}}$ messages of length $J+1$ and $2M\card{{E}}$ messages of length $1$. 
 The final application of the adjoint operator $\tilde{W}^*$ to recover the denoised signal estimates requires another $2M\card{{E}}$ messages, each a vector of length $J+1$. Therefore, the Chebyshev polynomial approximation to the spectral graph wavelet transform enables us to iteratively solve the weighted lasso 
in a distributed manner where the communication workload only scales with the  size of the network through $\card{E}$, and is otherwise independent of the network dimension $N$.

\section{Concluding Remarks and Future Work} \label{Se:conclusion}

We presented a novel method to distribute a class of linear operators called 
unions of graph Fourier multiplier operators. The main idea is to approximate the graph Fourier multipliers by Chebyshev polynomials, whose recurrence relations make them readily amenable to distributed computation in a sensor network. Key takeaways from the discussion and application examples include:
\begin{itemize}
\item A number of distributed signal processing tasks can be represented as distributed applications of unions of graph Fourier multiplier operators (and their adjoints) to signals on weighted graphs. Examples 
    include distributed smoothing, denoising, and semi-supervised learning.
\item The graph Fourier multiplier operators are the graph analog of filter banks, as they reshape functions' frequencies through multiplication in the Fourier domain.
\item The amount of communication required to perform the distributed computations only scales with the size of the network through the number of edges of the communication graph, which is usually sparse. Therefore, the method is well suited to large-scale sensor networks.
\end{itemize}

Our ongoing work includes extending the scope and depth of our application examples. In addition to considering more applications and larger size networks, we plan a more thorough empirical comparison of the computation and communication requirements of the approach described in this paper to alternative distributed optimization methods. The second major line of ongoing work is to analyze robustness issues that arise in real networks. For instance, we would like to incorporate quantization and communication noise into the sensor network model, in order to see how these propagate when using the Chebyshev polynomial approximation approach to distributed signal processing tasks. It is also important to analyze the effects of a sensor node dropping out of the network or communicating nodes losing synchronicity to ensure that the proposed method is stable to these disturbances.

\bibliographystyle{IEEEtran}
\bibliography{distributed}

\begin{thebibliography}{10}
\providecommand{\url}[1]{#1}
\csname url@samestyle\endcsname
\providecommand{\newblock}{\relax}
\providecommand{\bibinfo}[2]{#2}
\providecommand{\BIBentrySTDinterwordspacing}{\spaceskip=0pt\relax}
\providecommand{\BIBentryALTinterwordstretchfactor}{4}
\providecommand{\BIBentryALTinterwordspacing}{\spaceskip=\fontdimen2\font plus
\BIBentryALTinterwordstretchfactor\fontdimen3\font minus
  \fontdimen4\font\relax}
\providecommand{\BIBforeignlanguage}[2]{{%
\expandafter\ifx\csname l@#1\endcsname\relax
\typeout{** WARNING: IEEEtran.bst: No hyphenation pattern has been}%
\typeout{** loaded for the language `#1'. Using the pattern for}%
\typeout{** the default language instead.}%
\else
\language=\csname l@#1\endcsname
\fi
#2}}
\providecommand{\BIBdecl}{\relax}
\BIBdecl

\bibitem{Rabbat}
M.~Rabbat and R.~Nowak, ``Distributed optimization in sensor networks,'' in
  \emph{Proc. Int. Symp. Inf. Process. Sensor Netw.}, Berkeley, CA, Apr. 2004,
  pp. 20--27.

\bibitem{predd}
J.~B. Predd, S.~R. Kulkarni, and H.~V. Poor, ``Distributed learning in wireless
  sensor networks,'' \emph{IEEE Signal Process. Mag.}, vol.~23, pp. 56--69,
  Jul. 2006.

\bibitem{olfati}
R.~Olfati-Saber, J.~Fax, and R.~Murray, ``Consensus and cooperation in
  networked multi-agent systems,'' \emph{Proc. IEEE}, vol.~95, no.~1, pp.
  215--233, Jan. 2007.

\bibitem{dimakis}
A.~G. Dimakis, S.~Kar, J.~M.~F. Moura, M.~G. Rabbat, and A.~Scaglione, ``Gossip
  algorithms for distributed signal processing,'' \emph{Proc. IEEE}, vol.~98,
  no.~11, pp. 1847--1864, Nov. 2010.

\bibitem{smola}
A.~J. Smola and R.~Kondor, ``Kernels and regularization on graphs,'' in
  \emph{Proc. Ann. Conf. Comp. Learn. Theory}, ser. Lect. Notes Comp. Sci.,
  B.~Sch{\"o}lkopf and M.~Warmuth, Eds.\hskip 1em plus 0.5em minus 0.4em\relax
  Springer, 2003, pp. 144--158.

\bibitem{zhou_scholkopf}
D.~Zhou and B.~Sch{\"o}lkopf, ``A regularization framework for learning from
  graph data.'' in \emph{Proc. ICML Workshop Stat. Relat. Learn. and Its
  Connections to Other Fields}, Jul. 2004, pp. 132--137.

\bibitem{reg_discrete}
------, ``Regularization on discrete spaces,'' in \emph{Pattern Recogn.}, ser.
  Lect. Notes Comp. Sci., W.~G. Kropatsch, R.~Sablatnig, and A.~Hanbury,
  Eds.\hskip 1em plus 0.5em minus 0.4em\relax Springer, 2005, vol. 3663, pp.
  361--368.

\bibitem{harmonic}
X.~Zhu and Z.~Ghahramani, ``Semi-supervised learning using {G}aussian fields
  and harmonic functions,'' in \emph{Proc. Int. Conf. Mach. Learn.},
  Washington, D.C., Aug. 2003, pp. 912--919.

\bibitem{belkin_matveeva}
M.~Belkin, I.~Matveeva, and P.~Niyogi, ``Regularization and semi-supervised
  learning on large graphs,'' in \emph{Learn. Theory}, ser. Lect. Notes Comp.
  Sci.\hskip 1em plus 0.5em minus 0.4em\relax Springer-Verlag, 2004, pp.
  624--638.

\bibitem{zhou_bousquet}
D.~Zhou, O.~Bousquet, T.~N. Lal, J.~Weston, and B.~Sch{\"o}lkopf, ``Learning
  with local and global consistency,'' in \emph{Adv. Neural Inf. Process.
  Syst.}, S.~Thrun, L.~Saul, and B.~Sch{\"o}lkopf, Eds.\hskip 1em plus 0.5em
  minus 0.4em\relax {MIT Press}, 2004, pp. 321--328.

\bibitem{bougleux}
S.~Bougleux, A.~Elmoataz, and M.~Melkemi, ``Discrete regularization on weighted
  graphs for image and mesh filtering,'' in \emph{Scale Space Var. Methods
  Comp. Vision}, ser. Lect. Notes Comp. Sci., F.~Sgallari, A.~Murli, and
  N.~Paragios, Eds.\hskip 1em plus 0.5em minus 0.4em\relax Springer, 2007, vol.
  4485, pp. 128--139.

\bibitem{elmoataz}
A.~Elmoataz, O.~Lezoray, and S.~Bougleux, ``Nonlocal discrete regularization on
  weighted graphs: a framework for image and manifold processing,'' \emph{IEEE
  Trans. Image Process.}, vol.~17, pp. 1047--1060, Jul. 2008.

\bibitem{hancock}
F.~Zhang and E.~R. Hancock, ``Graph spectral image smoothing using the heat
  kernel,'' \emph{Pattern Recogn.}, vol.~41, pp. 3328--3342, Nov. 2008.

\bibitem{peyre_nlr}
G.~Peyr{\'e}, S.~Bougleux, and L.~Cohen, ``Non-local regularization of inverse
  problems,'' in \emph{Proc. ECCV'08}, ser. Lect. Notes Comp. Sci., D.~A.
  Forsyth, P.~H.~S. Torr, and A.~Zisserman, Eds.\hskip 1em plus 0.5em minus
  0.4em\relax Springer, 2008, pp. 57--68.

\bibitem{wagner}
R.~Wagner, V.~Delouille, and R.~Baraniuk, ``Distributed wavelet de-noising for
  sensor networks,'' in \emph{Proc. IEEE Int. Conf. Dec. and Contr.}, San
  Diego, CA, Dec. 2006, pp. 373--379.

\bibitem{barbarossa}
S.~Barbarossa, G.~Scutari, and T.~Battisti, ``Distributed signal subspace
  projection algorithms with maximum convergence rate for sensor networks with
  topological constraints,'' in \emph{Proc. IEEE Int. Conf. Acc., Speech, and
  Signal Process.}, Taipei, Apr. 2009, pp. 2893--2896.

\bibitem{guestrin}
C.~Guestrin, P.~Bodik, R.~Thibaux, M.~Paskin, and S.~Madden, ``Distributed
  regression: an efficient framework for modeling sensor network data,'' in
  \emph{Proc. Int. Symp. Inf. Process. Sensor Netw.}, Berkeley, CA, Apr. 2004,
  pp. 1--10.

\bibitem{predd_tit}
J.~B. Predd, S.~R. Kulkarni, and H.~V. Poor, ``A collaborative training
  algorithm for distributed learning,'' \emph{IEEE Trans. Inf. Theory},
  vol.~55, no.~4, pp. 1856--1871, Apr. 2009.

\bibitem{dlasso}
G.~Mateos, J.-A. Bazerque, and G.~B. Giannakis, ``Distributed sparse linear
  regression,'' \emph{IEEE Trans. Signal Process.}, vol.~58, no.~10, pp.
  5262--5276, Oct. 2010.

\bibitem{LTS-ARTICLE-2009-053}
D.~K. Hammond, P.~Vandergheynst, and R.~Gribonval, ``Wavelets on graphs via
  spectral graph theory,'' \emph{Appl. Comput. Harmon. Anal.}, vol.~30, no.~2,
  pp. 129--150, Mar. 2011.

\bibitem{chung}
F.~K. Chung, \emph{Spectral Graph Theory}.\hskip 1em plus 0.5em minus
  0.4em\relax Vol. 92 of the {CBMS} Regional Conference Series in Mathematics,
  {AMS} Bokstore, 1997.

\bibitem{lap_eigen}
T.~B{\i}y{\i}ko{\u{g}}lu, J.~Leydold, and P.~F. Stadler, \emph{Laplacian
  Eigenvectors of Graphs}.\hskip 1em plus 0.5em minus 0.4em\relax Lecture Notes
  in Mathematics, vol. 1915, Springer, 2007.

\bibitem{handscomb}
J.~C. Mason and D.~C. Handscomb, \emph{Chebyshev Polynomials}.\hskip 1em plus
  0.5em minus 0.4em\relax Chapman and Hall, 2003.

\bibitem{phillips}
G.~M. Phillips, \emph{Interpolation and Approximation by Polynomials}.\hskip
  1em plus 0.5em minus 0.4em\relax {CMS} Books in Mathematics, Springer-Verlag,
  2003.

\bibitem{rivlin}
T.~J. Rivlin, \emph{Chebyshev Polynomials}.\hskip 1em plus 0.5em minus
  0.4em\relax Wiley-Interscience, 1990.

\bibitem{anderson_morley}
W.~N. Anderson and T.~D. Morley, ``Eigenvalues of the {L}aplacian of a graph,''
  \emph{Linear Multilinear Algebra}, vol.~18, no.~2, pp. 141--145, 1985.

\bibitem{lasso}
R.~Tibshirani, ``Regression shrinkage and selection via the {L}asso,'' \emph{J.
  Royal. Statist. Soc. B}, vol.~58, no.~1, pp. 267--288, 1996.

\bibitem{basispursuit}
S.~Chen, D.~Donoho, and M.~Saunders, ``Atomic decomposition by basis pursuit,''
  \emph{SIAM J. Sci. Comp.}, vol.~20, no.~1, pp. 33--61, Aug. 1998.

\bibitem{DDD}
I.~Daubechies, M.~Defrise, and C.~De~Mol, ``An iterative thresholding algorithm
  for linear inverse problems with a sparsity constraint,'' \emph{Commun. Pure
  Appl. Math.}, vol.~57, no.~11, pp. 1413--1457, Nov. 2004.

\bibitem{combettes}
P.~L. Combettes and V.~R. Wajs, ``Signal recovery by proximal forward-backward
  splitting,'' \emph{Multiscale Model. Sim.}, vol.~4, no.~4, pp. 1168--1200,
  Nov. 2005.

\bibitem{par_dist_opt_book}
D.~P. Bertsekas and J.~N. Tsitsiklis, \emph{Parallel and Distributed
  Computation: {N}umerical Methods}.\hskip 1em plus 0.5em minus 0.4em\relax
  Prentice-Hall, 1989.

\end{thebibliography}

\end{document}